\documentclass[journal,onecolumn,12pt]{IEEEtran}
\usepackage{amsfonts,color,morefloats}
\usepackage{amssymb,amsmath,latexsym,amsthm}

\newtheorem{theorem}{Theorem}
\newtheorem{remark}{Remark}
\newtheorem{lemma}{Lemma}

\newtheorem{definition}{Definition}

\begin{document}

\title{Determination of the Autocorrelation Distribution and 2-Adic Complexity of Generalized Cyclotomic Binary Sequences of Order 2 with Period $pq$\thanks{The work was supported by the State Key Program of National Natural Science of China under Grant 12031011 and the National Natural Science Foundation of China under Grant 11701553.}}

\author{Xiaoyan Jing, \thanks{Xiaoyan Jing is with the Research Center for Number Theory and Its Applications, Northwest University, Xi'an 710127, China (Email: jxymg@126.com).}
Shiyuan Qiang,\thanks{Shiyuan Qiang is with the department of Applied Mathematics, China Agricultural University, Beijing 100083, China (E-mail: qsycau\_18@163.com).}  Minghui Yang\thanks{Minghui Yang is with the State Key Laboratory of Information Security, Institute of Information Engineering, Chinese Academy of Sciences, Beijing 100093, China (Email: yangminghui6688@163.com).} and
Keqin Feng \thanks{Keqin Feng is with the department of Mathematical Sciences, Tsinghua University, Beijing, 100084, China (Email: fengkq@mail.tsinghua.edu.cn).}
}

\date{\today}
\maketitle

\begin{abstract}
The generalized cyclotomic binary sequences $S=S(a, b, c)$ with period $n=pq$ have good autocorrelation property where $(a, b, c)\in \{0, 1\}^3$ and $p, q$ are distinct odd primes. For some cases, the sequences $S$ have ideal or optimal autocorrelation. In this paper we determine the autocorrelation distribution and 2-adic complexity of the sequences $S=S(a, b, c)$ for all $(a, b, c)\in \{0, 1\}^3$ in a unified way by using group ring language and a version of quadratic Gauss sums valued in group ring $R=\mathbb{Z}[\Gamma]$ where $\Gamma$ is a cyclic group of order $n$.
\end{abstract}


\section{Introduction} \label{sec-intro}
~\\

Let $p$ and $q$ be two distinct odd primes, $Z_{pq}=\mathbb{Z}/ pq\mathbb{Z}$. The unit group of the ring $Z_{pq}$ is
\begin{eqnarray*}
Z_{pq}^{*}
&=&  \left\{ a \pmod {pq}: \gcd(a, pq)=1\right\}\\
&=& \left\{ ip+jq \pmod {pq}: 1\leq i \leq q-1, 1\leq j \leq p-1\right\}.
\end{eqnarray*}

Let $P=\{p, 2p, \ldots, (q-1)p\}$, $Q=\{q, 2q, \ldots, (p-1)q\}$. Then $Z_{pq}=Z_{pq}^{\ast}\bigcup P \bigcup Q \bigcup \{0\}$.
The generalized cyclotomic binary sequence of order 2 with period $n=pq$ is $S(a, b, c)=S=\{s_\lambda\}_{\lambda=0}^{n-1}$, $s_\lambda\in \mathbb{F}_2=\{0, 1\}$ defined by
\begin{eqnarray}\label{sequence}
s_\lambda=\left\{ \begin{array}{ll}
c,                             & \mbox{ if $\lambda=0$} \\
a, & \mbox{ if $\lambda\in P$} \\
b, & \mbox{ if $\lambda\in Q$} \\
\frac{1}{2}\big(1-\big(\frac{\lambda}{p}\big)\big(\frac{\lambda}{q}\big)\big), & \mbox{ if $\lambda\in Z_{pq}^{*}$ }
\end{array}
\right.
\end{eqnarray}
where $(\frac{\cdot}{\cdot})$ denotes the Legendre symbol and $a, b, c\in \mathbb{F}_2.$ The autocorrelation value of $S=S(a, b, c)$ is defined by
$$C_S(\tau)=\sum_{\lambda=0}^{n-1}(-1)^{s_\lambda+s_{\lambda+\tau}}\in \mathbb{Z}~(0\leq\lambda\leq n-1)$$
and $C_S(0)=n$ is called the trivial autocorrelation of $S$. For many communication applications such as radar distance range, CDMA communication systems and generating key stream sequences in stream cipher encryption, we need the binary sequence $S$ having small absolute values $|C_S(\tau)|$ for all nontrivial autocorrelation
$(1\leq\tau\leq n-1)$. The linear feedback shift register (LFSR) and feedback with carry shift register (FCSR) are two kinds of fast key stream generators. The linear complexity (resp. 2-adic complexity) of the sequence measures the length of the shortest LFSR (resp. FCSR) which can generate the sequence. According to the Berlekamp-Massey algorithm \cite{M} (resp. rational approximation algorithm  \cite{K1}), the linear complexity (resp. 2-adic complexity) of a safe sequence should exceed half of its period.

The autocorrelation values of $S=S(a, b, c)$ have been computed in [3, 11] for several cases. Particularly, $S=S(a, b, c)$ has ideal autocorrelation value for
$q=p+2$ and $(a, b, c)=(1, 0, 0)$ or $(0, 1, 1)$, namely, $C_S(\tau)=-1$ for all $1\leq \tau\leq pq-1$, and has optimal $3$-valued autocorrelation for $q=p+4$
and $(a, b, c)=(1, 0, 0)$ or $(0, 1, 1)$, namely, $C_S(\tau)=1$ or $-3$ for $1\leq\tau\leq pq-1.$

 The linear complexity of the binary sequence $S=S(a, b, c)$ has been determined in [1, 4]. In this paper, our main aim is to determine the 2-adic complexity of the binary sequence $S=S(a, b, c)$. But before doing this we compute the autocorrelation distribution of $S$ in Section \ref{sec-auto}. This arrangement is made for several reasons. Firstly, we get the autocorrelation distribution of $S=S(a, b, c)$ for all cases $(a, b, c)\in \mathbb{F}_2^3$ by a unified method. The result shows that for all cases of $(a, b, c)$, the absolute values of all nontrivial autocorrelation  $|C_S(\tau)|~ (1\leq\tau\leq pq-1)$ are at most max$\{|q-p|+3, 9\}$ which means that the sequences $S=S(a, b, c)$ have nice autocorrelation property provided $|q-p|$ is small. Secondly, we use the group ring $R=\mathbb{Z}[\Gamma]$ language where $\Gamma=\langle x|x^n-1\rangle=\{1_\Gamma, x, \ldots, x^{n-1}\}$ is the multiplicative cyclic group of order $n=pq$
 generated by $x$ which makes the computation clearer and easier than the previous works [2, 3]. Thirdly, we use the same group ring method to determine the exact value of the 2-adic complexity of $S=S(a, b, c)$ for all $(a, b, c)\in \mathbb{F}_2^3$ in Section \ref{complexity}, comparing to previous works [5, 10, 13] where the authors use the methods given in [6, 14] and other tools as cyclotomy to get a lower bound of the 2-adic complexity of $S$ in general case. At last, we believe that the group ring method can also be used to determine the linear complexity of  $S=S(a, b, c)$ in a unified way for all  $(a, b, c)\in \mathbb{F}_2^3$ by using group ring $\mathbb{F}_2[\Gamma]$ instead of $\mathbb{Z}[\Gamma]$, but need more computation to determine if $S(\omega)$ is zero where $S(x)=\sum_{\lambda=0}^{pq-1}s_\lambda x^\lambda\in \mathbb{F}_2[x]$ and $\omega$ pasts through all $pq$-th roots of unity in the extension field of $\mathbb{F}_2$. The result would be complicated to state as shown in \cite{D1}, we do not touch it in this paper.

 \section{Autocorrelation Distribution}\label{sec-auto}
~\\

In this section we determine the autocorrelation distribution
$$C_s(\tau)\ (0\leq\tau\leq n-1)$$
of the binary sequences $S=S(a, b, c)$ for all $(a, b, c)\in \mathbb{F}_2^3$ with period $n=pq$ defined by (\ref{sequence}) in Section \ref{sec-intro}.

We use the group ring $R=\mathbb{Z}[\Gamma]$ where $\Gamma=\langle x|x^n-1\rangle=\{1_\Gamma, x, x^2, \ldots, x^{n-1}\}$ is the multiplicative cyclic group of order $n=pq$ generated by $x$. For a subset $S$ of $\Gamma$, we identify $S$ as an element $\sum_{g\in S}g$ in the group ring $\mathbb{Z}[\Gamma].$ Let
$$S(x)=S_{a, b, c}(x)=\sum_{\lambda=0}^{n-1}(-1)^{s_\lambda}x^\lambda\in R,$$
we have
\begin{align*}
S(x^{-1})S(x)&=\sum_{\lambda, \mu=0}^{n-1}(-1)^{s_\lambda+s_\mu}x^{\lambda-\mu}=\sum_{\tau=0}^{n-1}x^\tau\sum_{\mu=0}^{n-1}(-1)^{s_{\mu+\tau}+s_\mu}\notag\\
&=\sum_{\tau=0}^{n-1}C_s(\tau)x^\tau\in R.
\end{align*}
On the other hand, by the definition (\ref{sequence}) of the binary sequence $S=S(a, b, c)$,
$$S(x)=(-1)^c\cdot 1_{\Gamma}+(-1)^a(\Gamma_p-1_\Gamma)+(-1)^b(\Gamma_q-1_\Gamma)+\sum_{\lambda\in Z_{pq}^{\ast}}(\frac{\lambda}{p})(\frac{\lambda}{q})x^\lambda$$
where $\Gamma_p=\sum_{i=0}^{q-1}x^{ip}$, $\Gamma_q=\sum_{j=0}^{p-1}x^{jq}$ can be identified as subgroups $P\bigcup\{1_\Gamma\}$ and  $Q\bigcup\{1_\Gamma\}$ of $\Gamma$, and remark that for $\lambda\in Z_{pq}^{\ast}$ and $s_\lambda=\frac{1}{2}(1-(\frac{\lambda}{p})(\frac{\lambda}{q}))$, then $(-1)^{s_\lambda}=(\frac{\lambda}{p})(\frac{\lambda}{q})$. The last summation is
$$\sum_{\lambda\in Z_{pq}^{\ast}}(\frac{\lambda}{p})(\frac{\lambda}{q})x^\lambda=\sum_{i=1}^{q-1}\sum_{j=1}^{p-1}(\frac{jq}{p})(\frac{ip}{q})x^{ip+jq}=G_p(x)G_q(x)$$
where $$G_p(x)=\sum_{j=1}^{p-1}(\frac{jq}{p})x^{jq}\in R,~G_q(x)=\sum_{i=1}^{q-1}(\frac{ip}{q})x^{ip}\in R.$$
Therefore
\begin{align}
S(x)&=e\cdot 1_\Gamma+(-1)^a\Gamma_p+(-1)^b\Gamma_q+G_p(x)G_q(x)\notag\\
&=H+G_p(x)G_q(x)\label{e2}
\end{align}
where $e=(-1)^c-(-1)^a-(-1)^b\in \mathbb{Z}$, $H=e+(-1)^a\Gamma_p+(-1)^b\Gamma_q.$

Since $\sigma: \Gamma\rightarrow\Gamma$, $\sigma(g)=g^{-1}$ is an isomorphism of the group $\Gamma$, we have
$\sigma(\Gamma_p)=\Gamma_p$, $\sigma(\Gamma_q)=\Gamma_q$
and
$$\sigma(G_p(x))=\sum_{j=1}^{p-1}(\frac{jq}{p})x^{-jq}=\sum_{j=1}^{p-1}(\frac{-jq}{p})x^{jq}=(\frac{-1}{p})G_p(x),$$
$$\sigma(G_q(x))=(\frac{-1}{q})G_q(x).$$
We get, by (\ref{e2}),
\begin{align}
S(x^{-1})=\sigma(S(x))=H+(\frac{-1}{p})(\frac{-1}{q})G_p(x)G_q(x)\label{e3}
\end{align}
In order to compute $S(x)S(x^{-1})$ in $R$, we use the following lemma which particularly shows that $G_p(x)$ and $G_q(x)$ can be viewed as versions of quadratic Gauss sums over $\mathbb{F}_p$ and $\mathbb{F}_q$ respectively valued in $R$.
\begin{lemma}\label{le1}
In the group ring $R=\mathbb{Z}[\Gamma]$,\\
(A). $G_p^2(x)=(\frac{-1}{p})(p\cdot 1_{\Gamma}-\Gamma_q)$, $G_q^2(x)=(\frac{-1}{q})(q\cdot 1_{\Gamma}-\Gamma_p)$.\\
(B). $\Gamma_pG_q(x)=\Gamma_qG_p(x)=0, \Gamma_p\Gamma_q=\Gamma$.
\end{lemma}
\begin{proof}
(A). From the definition of $G_p(x)$ we have
\begin{align*}
G_p^2(x)&=\sum_{\lambda, \mu=1}^{p-1}(\frac{\lambda\mu}{p})x^{q(\lambda+\mu)} ~(\rm{let}~ \lambda=\mu\tau)\\
&=\sum_{\tau=1}^{p-1}(\frac{\tau}{p})\sum_{\mu=1}^{p-1}x^{q\mu(\tau+1)}\\
&=(\frac{-1}{p})(p-1)\cdot 1_\Gamma+\sum_{\tau=1}^{p-2}(\frac{\tau}{p})\sum_{\mu=1}^{p-1}x^{q\mu}\\
&=(\frac{-1}{p})(p-1)\cdot 1_\Gamma+(-(\frac{-1}{p}))(\Gamma_q-1_\Gamma)\\
&=(\frac{-1}{p})(p\cdot1_\Gamma-\Gamma_q)
\end{align*}
Similarly we have $G_q^2(x)=(\frac{-1}{q})(q\cdot 1_\Gamma-\Gamma_p).$

(B). Since $n=pq$ is a product of distinct primes $p$ and  $q$, we get $\Gamma_p\Gamma_q=\Gamma_p\times \Gamma_q$ (direct product)$=\Gamma$. From $\Gamma_p=\sum_{i=0}^{q-1}x^{pi}$ we have $x^{pi}\Gamma_p=\Gamma_p$ for each $i$. Therefore
$$\Gamma_pG_q=\sum_{i=1}^{q-1}(\frac{ip}{q})x^{pi}\Gamma_p=\sum_{i=1}^{q-1}(\frac{ip}{q})\Gamma_p=0.$$
Similarly we have $\Gamma_qG_p=0.$
\end{proof}
By Lemma \ref{le1} and (\ref{e2}), (\ref{e3}), we get
\begin{align}
{S(x^{-1})}S(x)&=(H+G_p(x)G_q(x))(H+(\frac{-1}{p})(\frac{-1}{q})G_p(x)G_q(x))\notag\\
&=H^2+(1+(\frac{-1}{p})(\frac{-1}{q}))HG_p(x)G_q(x)+(p\cdot1_\Gamma-\Gamma_q)(q\cdot1_\Gamma-\Gamma_p)\label{e4}
\end{align}
and
\begin{align*}
H^2&=(e\cdot 1_\Gamma+(-1)^a\Gamma_p+(-1)^b\Gamma_q)^2~ (e=(-1)^c-(-1)^a-(-1)^b)\\
&=e^2\cdot1_\Gamma+\Gamma_p\Gamma_p+\Gamma_q\Gamma_q+2(e(-1)^a\Gamma_p+e(-1)^b\Gamma_q+(-1)^{a+b}\Gamma_p\Gamma_q)\\
&=e^2\cdot1_\Gamma+q\Gamma_p+p\Gamma_q+2(e(-1)^a\Gamma_p+e(-1)^b\Gamma_q+(-1)^{a+b}\Gamma)~(\Gamma_p\Gamma_p=|\Gamma_p|\cdot\Gamma_p=q\Gamma_p)
\end{align*}
$$HG_p(x)G_q(x)=eG_p(x)G_q(x)$$
$$(p\cdot 1_\Gamma-\Gamma_q)(q\cdot1_\Gamma-\Gamma_p)=pq\cdot1_\Gamma-p\Gamma_p-q\Gamma_q+\Gamma$$
Then formula (\ref{e4}) becomes
\begin{align}
S(x^{-1})S(x)&=(pq+e^2)\cdot 1_\Gamma+(q-p+2e(-1)^a)\Gamma_p+(p-q+2e(-1)^b)\Gamma_q\notag\\
&~~~+(1+2(-1)^{a+b})\Gamma+e(1+(\frac{-1}{p})(\frac{-1}{q}))G_p(x)G_q(x)
\end{align}
Since $G_p(x)G_q(x)=\sum_{\lambda\in Z_{pq}^{\ast}}(\frac{\lambda}{p})(\frac{\lambda}{q})x^\lambda$, $\sum_{\lambda\in P}x^\lambda=\Gamma_p-1_\Gamma$, $\sum_{\lambda\in Q}x^\lambda=\Gamma_q-1_\Gamma$ and
$$\Gamma=1_\Gamma+(\Gamma_p-1_\Gamma)+(\Gamma_q-1_\Gamma)+\Gamma^{\ast}$$
where $\Gamma^{\ast}=\sum_{\lambda\in Z_{pq}^{\ast}}x^\lambda$, we get
\begin{align*}
S(x^{-1})S(x)&=pq+(q-p+2e(-1)^a)(\Gamma_p-1_\Gamma)+(p-q+2e(-1)^b)(\Gamma_q-1_\Gamma)\\
&~~~+(1+2(-1)^{a+b})[(\Gamma_p-1_\Gamma)+(\Gamma_q-1_\Gamma)+\Gamma^{\ast}]+e(1+(\frac{-1}{p})(\frac{-1}{q}))G_p(x)G_q(x)\\
&=pq+(q-p+2(-1)^{a+c}-1)(\Gamma_p-1_\Gamma)+(p-q+2(-1)^{b+c}-1)(\Gamma_q-1_\Gamma)\\
&~~~+\sum_{\lambda\in Z_{pq}^{\ast}}[e(1+(\frac{-1}{p})(\frac{-1}{q}))(\frac{\lambda}{p})(\frac{\lambda}{q})+1+2(-1)^{a+b}]x^\lambda.
\end{align*}

Therefore we get the following result on the autocorrelation distribution of the binary sequences $S=S(a, b, c).$
\begin{theorem}\label{th1}
Let $(a, b, c)\in\mathbb{F}_2^3$ and $S=S(a, b, c)$ be the binary sequence defined by (\ref{sequence}) with period $n=pq$. Then for $0\leq\tau\leq n-1$,
\begin{eqnarray*}
C_S(\tau)=\left\{ \begin{array}{ll}
pq,                             & \mbox{for $\tau=0$} \\
(q-p)+2(-1)^{a+c}-1, & \mbox{for $\tau\in P$} \\
(p-q)+2(-1)^{b+c}-1, & \mbox{for $\tau\in Q$}\\
1+2(-1)^{a+b}+((-1)^c-(-1)^a-(-1)^b)(1+(\frac{-1}{p})(\frac{-1}{q}))(\frac{\tau}{p})(\frac{\tau}{q}), & \mbox{ for $\tau\in Z_{pq}^{*}$ }
\end{array}
\right.
\end{eqnarray*}
\end{theorem}
\begin{remark}
(A) If $q-p\equiv 2\pmod 4$, then $n=pq\equiv p(p+2)\equiv 3\pmod 4$ and  $(\frac{-1}{p})(\frac{-1}{q})=-1.$ By Theorem \ref{th1}, the nontrivial autocorrelation of $S=S(a, b, c)$ takes values $q-p+2(-1)^{a+c}-1$, $p-q+2(-1)^{b+c}-1$ and $1+2(-1)^{a+b}$. Particularly,

(A.1) If $q=p+2$, and $(a, b, c)=(1, 0, 0)$ or $(0, 1, 1)$, the sequence $S=S(a, b, c)$ has ideal autocorrelation which means that $C_S(\tau)=-1$ for
all $1\leq\tau\leq n-1.$

(A.2) If $q=p+2$ and $(a, b, c)=(0, 0, 0)$ or $(1, 1, 1)$, the sequence $S=S(a, b, c)$ has autocorrelation $C_S(\tau)=-1$ or
$3$ for $1\leq\tau\leq n-1$.

(B) If $q-p\equiv 0\pmod 4$, then $n=pq\equiv 1\pmod 4$ and $(\frac{-1}{p})(\frac{-1}{q})=1$. By Theorem \ref{th1}, the nontrivial autocorrelation of
$S=S(a, b, c)$ takes values $q-p+2(-1)^{a+c}-1$, $p-q+2(-1)^{b+c}-1$ and $1+2(-1)^{a+b}\pm2((-1)^c-(-1)^a-(-1)^b)$. Particularly, if $q=p+4$ and $(a, b, c)=(1, 0, 0)$ or $(0, 1, 1)$, the sequence $S=S(a, b, c)$ has optimal autocorrelation which means that $C_S(\tau)=1$ or $-3$ for $1\leq\tau\leq n-1.$
\end{remark}
 \section{2-Adic Complexity}\label{complexity}
~\\

In this section we determine the exact value of the 2-adic complexity of binary sequences $S=S(a, b, c)$ for all $(a, b, c)\in \{0, 1\}^3.$
\begin{definition}\cite{K}
For a binary sequence $S=\{s_\lambda\}_{\lambda=0}^{n-1}$ with period  $n$, and $s_\lambda\in\{0, 1\}$. Let $T(2)=\sum_{\lambda=0}^{n-1}s_\lambda 2^\lambda\in \mathbb{Z}$. The 2-adic complexity of $S$ is defined by
$$A_S(2)=\log_2\frac{2^n-1}{d}$$
where $d=\gcd(T(2), 2^n-1)$.
\end{definition}
Now we consider $S=S(a, b, c)$. Let $T(x)=\sum_{\lambda=0}^{n-1}s_\lambda x^\lambda\in R$. Then
$$2T(x)=\sum_{\lambda=0}^{n-1}(1-(-1)^{s_\lambda})x^\lambda=\sum_{\lambda=0}^{n-1}x^\lambda-\sum_{\lambda=0}^{n-1}(-1)^{s_\lambda}x^\lambda=\sum_{\lambda=0}^{n-1}x^\lambda-S(x)$$
The above equality in $R=\mathbb{Z}[\Gamma]$ can be written as the following form
$$2T(x)\equiv \sum_{\lambda=0}^{n-1}x^\lambda-S(x)\equiv-S(x)+\frac{x^n-1}{x-1}\pmod{x^n-1}~(n=pq)$$
By taking $x=2$, we get
$$2T(2)\equiv -S(2)+2^n-1\equiv -S(2)\pmod{2^n-1}$$
Therefore $d=\gcd(T(2), 2^n-1)=\gcd(S(2), 2^n-1)$. For the equality (\ref{e2}) in $R$, by taking $x=2$ we get
\begin{align}
S(2)&\equiv e+(-1)^b\sum_{j=0}^{p-1}2^{qj}+(-1)^a\sum_{i=0}^{q-1}2^{pi}+G_p(2)G_q(2)\pmod {2^n-1}\notag\\
&\equiv e+(-1)^b\frac{2^n-1}{2^q-1}+(-1)^a\frac{2^n-1}{2^p-1}+G_p(2)G_q(2)\pmod {2^n-1} \label{e6}
\end{align}
where $e=(-1)^c-(-1)^a-(-1)^b$ and
$$G_p(2)=\sum_{j=1}^{p-1}(\frac{qj}{p})2^{qj},~G_q(2)=\sum_{i=1}^{q-1}(\frac{pi}{q})2^{pi}\in \mathbb{Z}.$$
Similarly, by taking $x=2$ in equality (\ref{e3}), we get
\begin{align}
S(2^{-1})\equiv e+(-1)^b\frac{2^n-1}{2^q-1}+(-1)^a\frac{2^n-1}{2^p-1}+(\frac{-1}{p})(\frac{-1}{q})G_p(2)G_q(2)\pmod{2^n-1} \label{e7}
\end{align}

Let $d_p=\gcd(S(2), 2^p-1)$, $d_q=\gcd(S(2), 2^q-1)$, $d^{\ast}=\gcd(S(2), \frac{2^n-1}{(2^p-1)(2^q-1)})$.
\begin{lemma}\label{le2}
(1).\ $d^*=1,\ d_p=\gcd(e+(-1)^aq,2^p-1),\ d_q=\gcd(e+(-1)^bq,2^q-1)$, where $e=(-1)^c-(-1)^a-(-1)^b.$

(2).\ If $4p>q+1$, then $d_p=1$.\ If $4q>p+1$, then $d_q=1$.

(3).\ $d(=\gcd(S(2),2^n-1))=\max(d_p,d_q)$.\ Particularly, if $16p>4q+4>p+5$, then $d=1$ and the binary sequence $S=S(a,b,c)$ has the best $2$-adic complexity $\log_2(2^{pq}-1).$
\end{lemma}

\begin{proof}
(1).\ Firstly we prove $d^*=1$.\ From (6) we have
$$S(2)\equiv e+G_p(2)G_q(2)\pmod{\frac{2^n-1}{(2^p-1)(2^q-1)}} \quad (n=pq)$$
Let $x=2$ in Lemma 1 (A) we get
$$G_p^2(2)\equiv(\frac{-1}{p})(p-\sum_{i=0}^{p-1}2^{qi})\equiv(\frac{-1}{p})(p-\frac{2^n-1}{2^q-1})\equiv(\frac{-1}{p})p \pmod{\frac{2^n-1}{(2^p-1)(2^q-1)}}$$
and similarly,
$$G_q^2(2)\equiv(\frac{-1}{q})q \pmod{\frac{2^n-1}{(2^p-1)(2^q-1)}}$$
Suppose that $d^*>1.$\ Let $\pi$ be a prime divisor of  $d^*=\gcd(S(2),\frac{2^n-1}{(2^p-1)(2^q-1)})$.\ Then $\pi\mid S(2)$ and $\pi\mid\frac{2^n-1}{(2^p-1)(2^q-1)}$.\ Therefore $2^{pq}\equiv1\pmod{\pi}$, so that the order of $2\pmod{\pi}$ is $p, q$ or $pq$, and
$$0\equiv S(2)\equiv e+G_p(2)G_q(2)\pmod{\pi} $$
Therefore
\begin{align}\label{e8}
e^2\equiv (G_p(2)G_q(2))^2\equiv(\frac{-1}{p})(\frac{-1}{q})pq\pmod{\pi}
\end{align}
which implies $\pi\mid pq\pm e^2$ where $e^2=1$ or $9$.

If $2^p\equiv1\pmod{\pi}$, by Fermat's Theorem we have $p\mid\pi-1$ so that $p\leq\frac{\pi-1}{2}$.\ Since $\pi$ and $p$ are odd primes. Moreover,
$$0\equiv \frac{2^{pq}-1}{(2^p-1)(2^q-1)}\equiv\frac{1}{2^q-1}\sum_{\lambda=0}^{q-1}2^{p\lambda}\equiv\frac{1}{2^q-1}\sum_{\lambda=0}^{q-1}1\equiv\frac{q}{2^q-1}\pmod{\pi}$$
we get $\pi=q$.\ Then by $\pi\mid pq\pm e^2$ we get $\pi\mid e$ so that $\pi=3$.\ On the other hand, $p, q$ and $n=pq$ are odd, we get a contradiction
$$0\equiv \frac{2^n-1}{(2^p-1)(2^q-1)}\equiv 1\pmod{3}$$
Therefore $2^p\not\equiv1\pmod{\pi}.$\ Similarly we have $2^q\not\equiv1\pmod{\pi}.$\ Therefore the order of $2\pmod{\pi}$ is $pq$.\ Then $pq\mid\pi-1$.\ From this and $\pi\mid pq\pm e^2$ we get $\pi\leq\frac{1}{2}(pq+9)$ and $pq\leq\frac{1}{2}(\pi-1).$\ Therefore $4\pi\leq2pq+18\leq2(\pi-1)+18=2\pi+16$, and $\pi\leq8$. From $\frac{1}{2}(\pi-1)\geq pq\geq15$ we get $\pi\geq31$, a contradiction.\ Therefore $d^*$ has no prime divisor.\ Namely, $d^*=1$.

Then we determine $d_p=\gcd(S(2),2^p-1).$\ From (6) we have
$$S(2)\equiv e+(-1)^a\frac{2^{pq}-1}{2^p-1}+G_p(2)G_q(2)\pmod{2^p-1}$$
Since
$$\frac{2^{pq}-1}{2^p-1}\equiv q,\quad G_q(2)=\sum_{i=1}^{q-1}(\frac{pi}{q})2^{pi}\equiv\sum_{i=1}^{q-1}(\frac{pi}{q})\equiv 0\pmod{2^p-1}$$
We get $S(2)\equiv e+(-1)^aq\pmod{2^p-1}$ and $d_p=\gcd(S(2),2^p-1)=\gcd(e+(-1)^aq,2^p-1)$.
Similarly we have $d_q=\gcd(e+(-1)^bp,2^q-1).$

(2).\ Suppose that $d_p>1.$ Let $\pi$ be a prime divisor of $d_p=\gcd(e+(-1)^aq,2^p-1)$.\ Then $p\mid\frac{\pi-1}{2}$, and $\pi\mid\frac{1}{2}(q+(-1)^ae)\leq\frac{1}{2}(q+3)$. Therefore $4p\leq2\pi-2\leq q+3-2=q+1$.\ This means that if $4p>q+1$ then $d_p=1$.\ Similarly, if $4q>p+1$ then $d_q=1$.

(3).\ From $d^*=1$ and $\gcd(2^p-1, 2^q-1)=1$ we get
$$d=\gcd(S(2),2^n-1)=\gcd(S(2),(2^p-1)(2^q-1))=\gcd(S(2),2^p-1)\cdot\gcd(S(2),2^q-1)=d_pd_q$$
Moreover, $4p\leq q+1$ and $4q\leq p+1$ can not be hold simultaneously.\ We know that $d_p=1$ or $d_q=1$.\ Therefore $d=\max(d_p,d_q)$.

\end{proof}

In summary we have the following result.

\begin{theorem}
Let $S=S(a,b,c)$ be the binary sequences with period $n=pq$ defined by (\ref{sequence}), $(a,b,c)\in\{0,1\}^3$.\ Then

(A).\ the $2$-adic complexity of $S$ is
$$A_s(2)=\log_2(\frac{2^n-1}{\max(d_p,d_q)})$$
where $d_p=\gcd(q-1+(-1)^{a+c}-(-1)^{a+b},2^p-1)$, $d_q=\gcd(p-1+(-1)^{b+c}-(-1)^{a+b},2^q-1)$.

(B).\ If $4p>q+1$, then $d_p=1$.\ If $4q>p+1$, then $d_q=1$.\ Particularly, if $16p>4q+4>p+5$, then the $2$-adic complexity of $S(a,b,c)$ reaches the best value $\log_2(2^{n}-1)$ for all $(a,b,c)\in\{0,1\}^3$.
\end{theorem}

\begin{remark}
At the end of Section \ref{sec-auto} we list the sequences $S=S(a,b,c)$ with ideal or optimal autocorrelation.\ For all of such sequences, $q=p+2$ or $p+4$, so $p$ and $q$ satisfy $16p>4q+4>p+5$.\ By Theorem 2 (B), the  $2$-adic complexity of all such sequences reaches the best value $\log_2(2^{pq}-1).$
\end{remark}

\end{document}